\documentclass[conference, 10pt]{IEEEtran}
\ifCLASSINFOpdf
\else
\fi
\usepackage{cite}
\usepackage{amsmath,graphicx,amssymb,color,textcomp,amsfonts,subfigure,enumerate,bm}
\usepackage{extarrows}
\usepackage{makecell,multirow,diagbox,stfloats}
\newcommand{\reffig}[1]{Fig. \ref{#1}}

\usepackage{amsthm}
\usepackage{algorithmicx,algorithm,balance}
\usepackage{algpseudocode}
\usepackage{colortbl}
\usepackage{array}
\usepackage{booktabs}
\usepackage{forloop}
\newcounter{ct}


\newtheorem{theorem}{Theorem}

\begin{document}
%\title{{Integrated Sensing and Communication for IoT Zero-Power Devices}}
\title{{Net-Zero} Integrated Sensing and Communication in Backscatter Systems}
\author{
	Yu Zhang, Tongyang Xu, Christos Masouros$^\dag$, and Zhu Han$^\ddag$\\
	\normalsize	School of Engineering, Newcastle University, Newcastle upon Tyne NE1 7RU, United Kingdom\\
	{$^\dag$Department of Electronic and Electrical Engineering, University College London, London WC1E 6BT, United Kingdom}\\
	$^\ddag$Department of Electrical and Computer Engineering, University of Houston, TX 77004, USA\\
	Email: yu.zhang@newcastle.ac.uk, tongyang.xu@newcastle.ac.uk, c.masouros@ucl.ac.uk, zhan2@uh.edu
\thanks{This work was supported in part by the xxx.}
}

\maketitle
\begin{abstract}
Future wireless networks targeted for improving spectral and energy efficiency, are expected to simultaneously provide sensing functionality and support low-power communications. This paper proposes a novel net-zero integrated sensing and communication (ISAC) model for backscatter systems, including an {access point (AP)}, a net-zero device, and a user receiver. We fully utilize the backscatter mechanism for sensing and communication without additional power consumption and signal processing in the hardware device, which reduces the system complexity and makes it feasible for practical applications. To further optimize the system performance, we design a novel signal frame structure for the ISAC model that effectively mitigates communication interference at the transmitter, tag, and receiver. Additionally, we employ distributed antennas for sensing which can be placed flexibly to capture a wider range of signals from diverse angles and distances, thereby improving the accuracy of sensing. We derive theoretical expressions for the symbol error rate (SER) and {tag's location detection probability}, and provide a detailed analysis of how the system parameters, such as transmit power and tag's reflection coefficient, affect the system performance.
\end{abstract}


\section{introduction}
Integrated sensing and communication (ISAC) has emerged as an important technology in the evolution towards 6G communications. The key concept of ISAC is integrating sensing and communication functions into one device, which improves the {spectrum and energy efficiency} and reduces the hardware cost \cite{OpenJ}. This dual functionality makes ISAC beneficial in applications such as autonomous vehicles, smart cities, and the Internet of Things (IoT), where real-time sensing and communication are essential. ISAC has gained significant research interest in {areas}, such as waveform design, and trade-off analysis \cite{full}.

%One of the most promising applications of ISAC is in the autonomous vehicles area, where it helps the vehicles to communicate with each other while simultaneously provides real-time sensing of nearby obstacles and road conditions. In the indoor environment, ISAC can utilize the existing communication infrastructure, such as Wi-Fi, to perform accurate localization and tracking. This application is particularly useful in complex indoor spaces where traditional GPS services are ineffective, which enables ISAC provide a cost-effective solution without the need for additional specialized equipment.

With the popularity of sustainable development and energy-saving concepts, {net-zero} devices have attracted much attention \cite{yu}. Unlike traditional devices that require batteries, net-zero devices operate without a battery. Instead, they harvest energy from the surrounding environment, particularly from ambient radio frequency (RF) signals. 
%The zero-power devices have an internal energy harvester circuit, typically consisting of a rectifier and a capacitor, which converts the RF waves into direct current (DC) power. The harvested energy powers the micro-controller, which controls the backscatter modulation by changing the load impedance. The converted DC power can be temporarily stored in a capacitor, which allows the zero-power device to have a small reserve of energy to maintain its operation in the case of the external RF signals are absent. 
The working mechanism of net-zero devices is backscattering, where they reflect incoming RF signals from nearby transmitters, such as TV towers, mobile phones, or WiFi routers, to convey information \cite{ambient}. This backscatter mechanism allows the devices to operate at extremely low power levels, often in the microwatt range. 
%The typical applications of zero-power device include industrial wireless sensors for environmental information collection and smart wearable devices for health data monitoring.

ISAC technology is particularly well-suited for net-zero devices since they can perform sensing and communication tasks simultaneously using the same RF signals. The net-zero devices can provide sensing information while simultaneously transmitting information data to a specific receiver. However, there is limited research conducted on {backscatter enabled ISAC} system. In \cite{ABC}, the authors introduce an ISAC framework including a base station (BS), a tag, and a user, which derives the communication rate of the user and the tag, as well as the sensing rate at the BS. Similar to conventional ISAC systems, this model still requires optimization of the power allocation factor to balance the trade-off between communication and sensing. Another study involves ISAC for reconfigurable intelligent surface (RIS)-assisted backscatter communication \cite{risassisted}, where the BS simultaneously detects backscattered signals from the RIS-assisted IoT devices and senses targets based on the echo signals. However, this study doesn't consider the interference between different IoT devices.

The above-mentioned literature on ISAC system design mainly focuses on the joint waveform design for sensing and communication. These studies do not fully leverage the potential of backscattered signals from net-zero devices, which can be used for both communication and sensing. In our study, we utilize an access point (AP), such as Wi-Fi, to send signals to net-zero devices. The backscattered signal from the net-zero device can be used for sensing purposes, such as localization. This approach eliminates the need for additional waveform designs for sensing and communication, thereby reducing system complexity and enhancing energy efficiency.

\section{System Model}

%\begin{figure}[ht]
%	\centering
%	% Requires \usepackage{graphicx}
%	\includegraphics[width=85mm]{3D Scenario model.pdf}\\
%	\caption{An ISAC scenario for the IoT zero-power devices.}\label{model}
%\end{figure}
\begin{figure}[t]
	\centering
	\includegraphics[width=80mm]{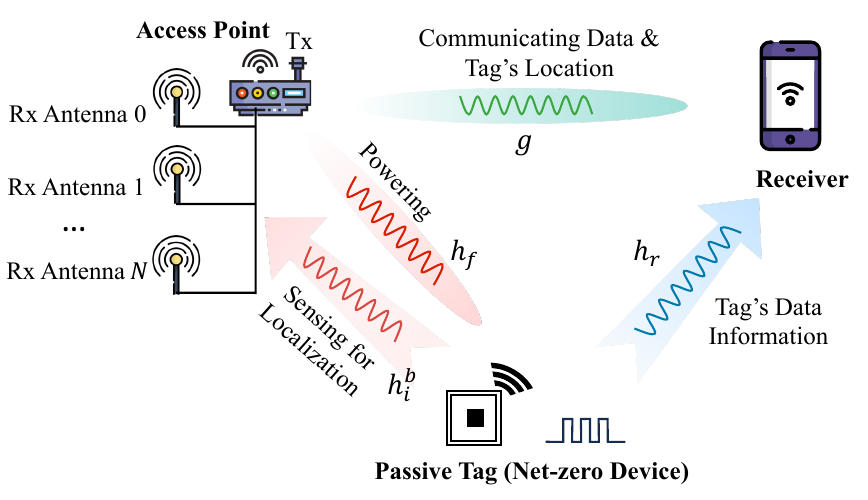}\\
	\caption{System model of the ISAC for net-zero devices.}\label{model3}
\end{figure}

In \reffig{model3}, the proposed net-zero ISAC system model consists of an AP, a passive tag, and a receiver. The {AP} is equipped with a single transmit antenna and $N$ distributed receive antennas. The advantage of using distributed antennas is that they can be placed across different indoor locations. This placement flexibility enables the AP to capture a wider range of signals from diverse angles and distances, thereby improving the accuracy of sensing. Both the transmitter and receiver antennas have omnidirectional radiation patterns. 

The AP sends a query signal to activate the passive tag. Subsequently, the tag modulates its data onto the incident RF signals and backscatters them toward both the AP and the receiver. The AP leverages the received backscattering signals across all $N$ distributed receive antennas for sensing. In this paper, we focus on using the sensing information for localization. The advantage of localization sensing lies in three aspects. Firstly, the sensing results assist {in beamforming for} powering the tag. Initially, an omnidirectional beam is used to both power the tag and sense its location. Once the location is determined, a directional beam is employed to enhance power delivery to the tag, simultaneously improving communication performance. Secondly, in indoor environments, localization can track user movements and detect intrusions. Thirdly, in underground spaces where GPS signals are unavailable, localization proves effective.

Since the proposed ISAC model is based on existing hardware, it does not require additional hardware costs. The signal processing for localization is conducted at the AP. The receiver, a low-cost device, captures both the signal from the AP and the backscattered signals from the tag. These backscattered signals contain the tag's data, such as temperature and humidity, which the receiver demodulates. Simultaneously, the signals from the AP provide communication services and the location information of the tag.

\subsection{Frame Structure Design}
To reduce the complexity of the proposed system, we design a novel frame structure of the ISAC model that effectively mitigates interference at the transmitter, tag, and receiver. As illustrated in \reffig{frame}, the frame structure contains a downlink (DL) time slot and an uplink (UL) time slot. During the DL time slot, the AP sends signals to power the tag while simultaneously transmitting communication data to the receiver. In contrast, during the UL time slot, the AP is configured to only receive signals from the tag without transmitting any signals itself. Consequently, during the UL time slot of the AP, the receiver only captures the backscattered signals from the tag, free from interference caused by transmissions from the AP. This strategic design of the frame structure effectively achieves interference elimination, reducing system complexity.

\begin{figure}[t]
	\centering
	\includegraphics[width=88mm]{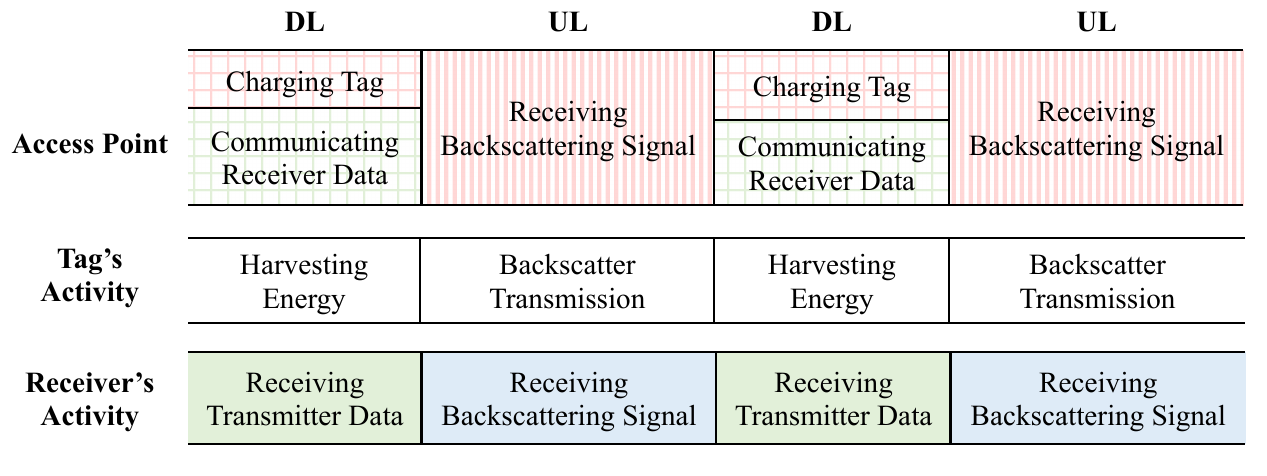}\\
	\caption{Frame structure of AP corresponding to the activities of the passive tag and receiver.}\label{frame}
\end{figure}
\begin{figure*}[t]
	\centering 
	\subfigure[Passive tag circuit including the antenna, energy harvester and modulator.]
	{
		\includegraphics[width=0.3\linewidth]{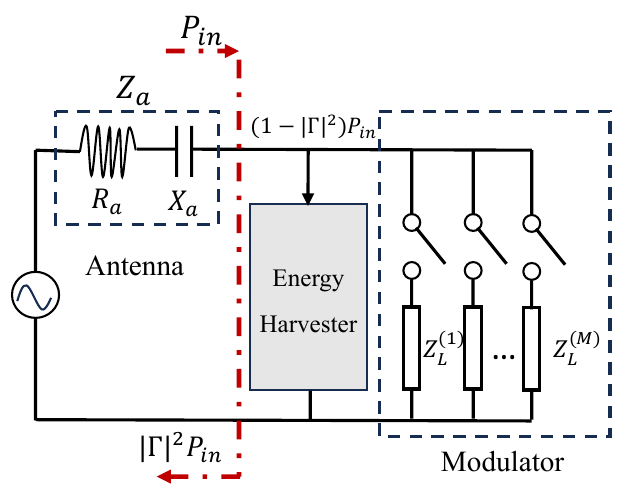}
		\label{tag}
	}
	\hspace{0.5cm}
	\subfigure[Impedance Smith chart: mapping the load impedance and reflection coefficient.]
	{
		\includegraphics[width=0.3\linewidth]{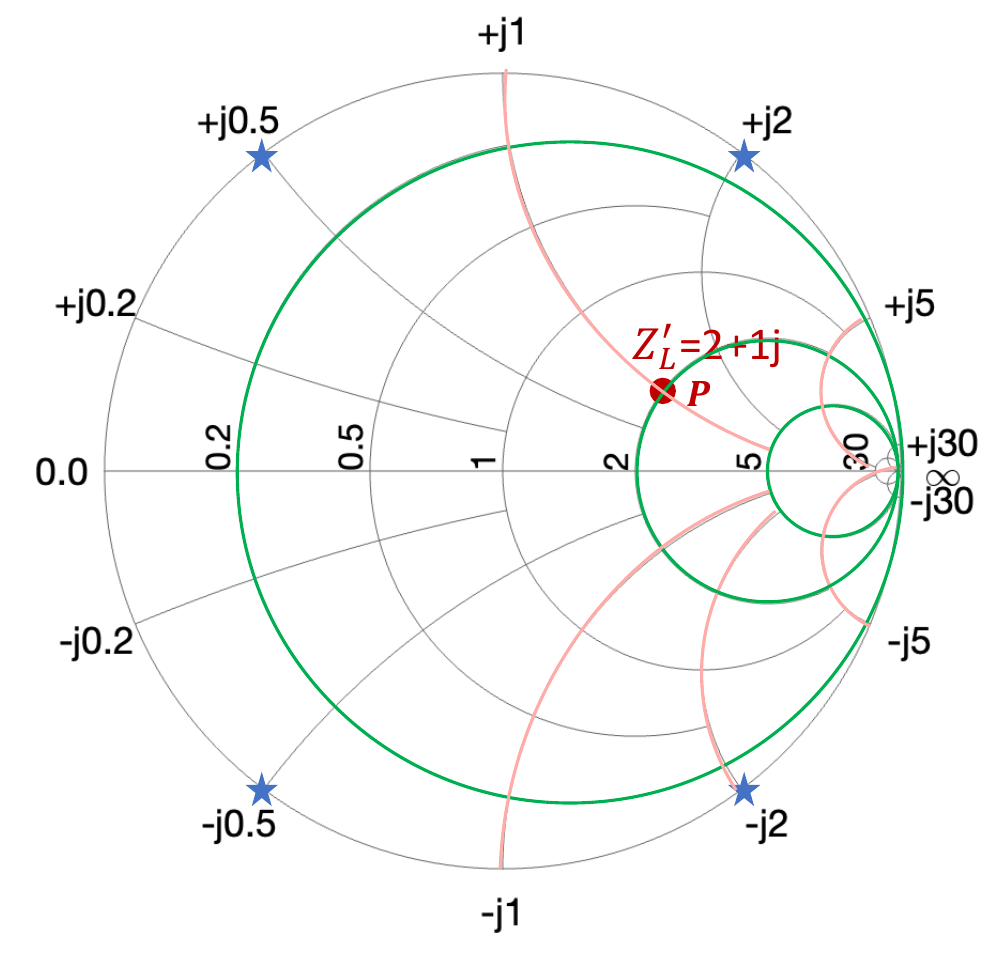}
		\label{smith}
	}
	\caption{The passive tag circuit and the Smith chart.}
	\label{modulation}
\end{figure*}

\subsection{Backscatter Modulation}
As a net-zero device, the passive tag doesn't have an internal battery to supply the power. Instead, it harvests the energy from the incident RF signals, modulating its own data onto the RF carrier and reflecting it back to the AP. The modulation of a tag is achieved by adjusting the reflection coefficient to realize different modulation schemes.
As shown in \reffig{tag}, to create an $M$-ary modulation, the tag utilizes $M$ load impedance values to generate $M$ different reflection coefficients, {with the $m$-th coefficient} defined as \cite{ambient}
\begin{equation}\label{coefficient}
\Gamma^{(m)}=\frac{Z_L^{(m)}-Z_a}{Z_L^{(m)}+Z_a}=\Gamma_R^{(m)}+j\Gamma_I^{(m)},
\end{equation}
where $Z_L=R_L+jX_L$ is the complex load impedance and $Z_a=R_a+jX_a$ is the complex antenna impedance. The real components, $R_L$ and $R_a$ are the load and antenna resistances, respectively, while the imaginary components, $X_L$, and $X_a$ are the load and antenna reactances, respectively. In \eqref{coefficient}, the reflection coefficient $\Gamma$ is a complex value comprising the real part $\Gamma_R$ and imaginary part $\Gamma_I$. This complex nature allows the modulation of both the amplitude and phase of the reflected signals by adjusting the load impedance $Z_L$.  As shown in \reffig{tag}, the incident RF signals captured by the passive tag's antenna, $|\Gamma|^2P_{in}$ of RF power is backscattered while $(1-|\Gamma|^2)P_{in}$ is absorbed in the load. 

The Smith chart is an effective tool for mapping load impedance to the reflection coefficient. To utilize the Smith chart, we need to first normalize the load impedance, i.e., $Z_L'=Z_L/Z_a$, and then $\Gamma$ is calculated as $\Gamma=\frac{Z_L'-1}{Z_L'+1}$.
As shown in \reffig{smith}, the Smith chart consists of two types of circles: the green circle represents the real part and the pink circle represents the imaginary part of $Z_L'$.
For example, given $Z_a=50\Omega$ and $Z_L=100+j50\Omega$, after normalizing, $Z_L'$ becomes $2+j$.  Plot the intersection point of the resistance circle with value 2 and the reactance circle with value 1j, i.e., point $P$ in \reffig{smith}. The projections on the horizontal and vertical axes of the Cartesian coordinates are the real part and the imaginary part of the reflection coefficient, respectively. Here we get point $P$ corresponding to the reflection coefficient $\Gamma=0.4+0.2j$. {For example}, by adjusting the value of $Z_L^{(m)}$, we can easily achieve 4-QAM modulation, as demonstrated by the four blue star points shown in \reffig{smith}.
%{\subsection{Energy Harvesting}
%As shown in \reffig{tag}, the incident RF signals captured by the passive tag's antenna, $|\Gamma|^2P_{in}$ of RF power is backscattered while $(1-|\Gamma|^2)P_{in}$ is absorbed at the load. The energy harvester, typically consisting of a rectifier and a capacitor, converts the RF waves into direct current (DC) power. The harvested energy powers the micro-controller of the tag, which is used for controlling the backscatter modulation by changing the load impedance. The converted DC power may be temporarily stored in a capacitor, which allows the tag to have a small reserve of energy to maintain its operation in the case of the external RF signals is not present. We use the linear model of energy harvesting, the harvested energy is given by
%\begin{align}
%P_h=\alpha(1-|\Gamma|^2)P_{in},
%\end{align}
%where $\alpha\in[0, 1]$ is power conversion efficiency of tag.}
\subsection{{AP} Sensing Model}
The received signal strength (RSS) {at AP which is backscattered} from the tag can be used as sensing information for localization. RSS measurements can be easily obtained from the received antennas without requiring additional hardware. Moreover, this approach is cost-effective and {low-power}, making it particularly suitable for net-zero device localization. The fundamental principle of the RSS-based localization method involves computing the distance between the transmitter and receiver by comparing the transmit and receive power.
The total received power at the $i$-th antenna is given by \cite{cascaded}
\begin{equation}\label{receivedpower}
{P_i^R =\chi P_{s}G_{tag}G_TG_R^iL(d_f)L(d_i^b) \left|h_f h^b_i\Gamma\right|^2,}
\end{equation}
where $\chi$ is the polarization loss factor, $P_s$ is the transmit power of AP, $G_{tag}$ is the antenna gain of the tag, $G_T$ and $G_R^i$  are the antenna gains of transmit antenna and $i$-th receive antenna of AP, respectively. The distance between transmit antenna and {$i$-th} receive antenna of AP to the tag are denoted as $d_f$ and $d_i^b$, respectively. $L(\cdot)$ is the channel path loss model, given by \cite{book} 
\begin{equation}
L(d_k) =\left( \frac{d_k}{d_0}\right)^{-\eta},
\end{equation}
where $\eta$ is the path loss exponent and $d_0$ is the reference distance. 

It is important to note that the received backscattered power must exceed the sensitivity threshold of the receiving antennas. Under this condition, the tag's location becomes detectable. For a given tag, if it is captured by receiving antenna $i$, we denote the \emph{antenna capture indicator} $L_i $ as 1, where it is defined as: \cite{IoTlocalization}
\begin{equation}
L_i=
\begin{cases}
1,& P_i^R\geq P_{th},\\
0,&\text{{otherwise}},
\end{cases}
\end{equation}
where $P_{th}$ is the sensitivity threshold. 
In the RSS-based trilateration localization method, determining a tag's location requires a minimum of three distinct RSS measurements. Consequently, the probability of successfully localizing a tag is defined by the following:
\begin{equation}\label{outageprobility}
P_{Loc} = \Pr\left\{\sum\limits_{i=1}^{N}L_i\geq 3\right\}.
\end{equation}
\section{Performance Analysis}
In this section, we will give the performance analysis of the proposed ISAC for net-zero devices, which is focused on how the transmit power of AP and reflection coefficient of tag affect the {tag's location detection probability} and symbol error rate (SER). Specifically, we will present a detailed derivation of the tag's location detection probability and theoretical upper bound of SER under different modulation schemes. This analysis aims to provide insights on how to adjust the transmit power of AP, the position of distributed antennas and reflection coefficient of tag to improve the system performance.
\subsection{Localization Performance}
The RSS-based localization method relies on the strength of the received signal from different positions, which is significantly influenced by multipath propagation. Therefore, choosing an appropriate channel fading model is very important. In this paper, we focus on the  Rician and Rayleigh fading models to represent indoor signal propagation. The Rician model is usually for scenarios where there is a dominant line-of-sight (LOS) component along with numerous indirect paths. The Rayleigh model is used in environments where there is no dominant line-of-sight path, and the signal is scattered in multiple directions.
If we consider $h_f$ and $h^b_i$ are independent Rician random variables, then the probability density function (PDF) of $h_i=h_fh^b_i$ is expressed as \cite{rician}
\begin{equation}
\begin{aligned}
p_{h_i}(x) &= \frac{x\exp\left[-(K_f+K_b^i)\right]}{\sigma_f^2\sigma_b^2}\!\!\sum\limits_{m=0}^{\infty}\sum\limits_{n=0}^{\infty}\frac{1}{(m!n!)^2}\!\!\left(\frac{xK_f}{2\sigma_f^2}\right)^m\\
&\times\!\left(\frac{xK_b^i}{2\sigma_b^2}\right)^n\!\!\left(\frac{\sigma_f}{\sigma_b}\right)^{m-n}\!\!\textbf{K}_{m-n}\!\!\left(\frac{x}{\sigma_f\sigma_b}\right), x\geq 0,
\end{aligned}
\end{equation}
where $\textbf{K}(\cdot)$ is the modified Bessel function of the second kind, $\{K_f, \sigma_f\}$ and $\{K_b^i,\sigma_b\}$ are the Rician shape parameters for both the forward and backscatter links, respectively, where $K_f$ and $K_b^i$ is defined as the ratio of the power contributions by line-of-sight (LoS) path to the remaining multipath. For the special case of $K_f=K_b^i=0$, the envelope of $h_f$ and $h_i^b$ become two independent Rayleigh random variables.

To derive the localization probability in \eqref{outageprobility}, we can use the binomial distribution to calculate the probability of at least three out of I receive antennas' RSS measurements exceeding the threshold $P_{th}$. Let's denote $X=\sum\limits_{i=1}^{N}L_i$, and then \eqref{outageprobility} can be rewritten as
\begin{equation}
\begin{aligned}
P_{Loc}&=1-\Pr\{X\leq3\}\\
&=1-\Pr\{X=0\}-\Pr\{X=1\}-\Pr\{X=2\},
\end{aligned}
\end{equation}
where $\Pr\{X=k\}=\sum\limits_{i=1}^k\left(\prod\limits_{j=1,j\neq i}^kp_j\right)(1-p_i)$ and
 $p_i$ is given as follows
\begin{equation}\label{pi}
\begin{aligned}
p_i &=\Pr\left\{P_i^R\geq P_{th}\right\}\\
&=\Pr\left\{\chi P_{s}G_{tag}G_TG_R^iL(d_f)L(d_i^b) \left|h_f h^b_i\Gamma\right|^2\geq P_{th}\right\},
\end{aligned}
\end{equation}
where $p_i$ represents the probability that $i$-th receive antenna of AP can detect the location of tag. Let's denote $\zeta=\chi P_{s}\left|G_{tag}G_TG_R^iL(d_f)L(d_i^b)\Gamma\right|^2$, then \eqref{pi} can be written as
\begin{equation}
\begin{aligned}
p_i&=1-\int_{0}^{\sqrt{P_{th}/\zeta}}p_{h_i}(x)dx\\
%&=1-\sum\limits_{m=0}^{\infty}\sum\limits_{n=0}^{\infty}\int_{0}^{\sqrt{P_{th}/\zeta}}Ax^{m+n+1}\textbf{K}_{m-n}\left(\frac{x}{\sigma_f\sigma_b}\right)dx,\notag\\
&=1\!\!-\!\!\frac{A}{2}\!\!\sum\limits_{m=0}^{\infty}\sum\limits_{n=0}^{\infty}\int_{0}^1\!\!\!\left(\frac{P_{th}}{\zeta}\right)^{\frac{m+n+2}{2}}\!\!\!\!\!\!x^{\frac{i+l}{2}}\textbf{K}_{m-n}\!\left(\!\sqrt{\frac{P_{th}x}{\zeta\sigma_f^2\sigma_b^2}}\right)\!\!dx,
\end{aligned}
\end{equation}
where $A= \frac{\exp\left[-(K_f+K_b^i)\right]}{\sigma_f^2\sigma_b^2(m!n!)^2}\left(\frac{K_f}{2\sigma_f^2}\right)^m\left(\frac{K_b^i}{2\sigma_b^2}\right)^n\left(\frac{\sigma_f}{\sigma_b}\right)^{m-n}$. According to the Eq. (6.592.2) in \cite{integralbook}, we can derive the tag location detection probability of antenna $i$ for the Rician fading scenario, as follows:
\begin{equation}\label{Ricianfading}
\begin{aligned}
p_i=1-\frac{A}{4}\sum\limits_{m=0}^{\infty}\sum\limits_{n=0}^{\infty}(2\sigma_f^2\sigma_b^2)^{m-n}\left(\frac{P_{th}}{\zeta}\right)^{n+1}\\
\times G^{2,1}_{1,3}\left(\left.\frac{P_{th}}{4\zeta\sigma_f^2\sigma_b^2}\right|_{m-n, 0, -n-1}^{-m}\right),
\end{aligned}
\end{equation}
where $G^{m,n}_{p,q}$ is the Meijer G-function, as defined in \cite[P1032, Eq. (9.301)]{integralbook}. From \eqref{Ricianfading}, we note that the derived tag location detection probability involves two infinite summation series, making it difficult to determine how the parameters affect system performance. To simplify the analysis and gain insights on parameter effects, we consider the scenario where $K_f=K_b^i=0$ in \eqref{Ricianfading}, corresponding to Rayleigh fading. The tag location detection probability under Rayleigh fading is then given by:
\begin{equation}
\begin{aligned}
p_i=1-\frac{P_{th}}{4\zeta\sigma_f^2\sigma_b^2}G^{2,1}_{1,3}\left(\left.\frac{P_{th}}{4\zeta\sigma_f^2\sigma_b^2}\right|_{0, 0, -1}^{0}\right).
\end{aligned}
\end{equation}
The derived $p_i$ under Rayleigh fading is mathematically less complex, involving only a single Meijer G-function. By examining the monotonicity of the function, we can directly obtain how parameters influence the system performance.
\begin{theorem}\label{theorem}
$f(x)=1-xG^{2,1}_{1,3}\left(x \Big|_{0, 0, -1}^{0}\right)$  is a monotonically decreasing when $x>0$.
\end{theorem}
\begin{proof}
According to \cite[Eq. (9.31.3)]{integralbook} , we can derive the first derivative of $f(x)$ is
\begin{equation}
\begin{aligned}
\frac{df(x)}{dx}&=-G^{2,1}_{1,3}\left(x \Big|_{0, 0, -1}^{0}\right)-x\frac{dG^{2,1}_{1,3}\left(x \Big|_{0, 0, -1}^{0}\right)}{dx}\\
&=-G^{2,1}_{1,3}\left(x \Big|_{0, 0, -1}^{-1}\right)=-2K_0(2\sqrt{x})<0.
\end{aligned}
\end{equation}
Since the first derivative of $f(x)<0$, $f(x)$ is a monotonically decreasing function. Therefore, the Theorem \ref{theorem} is proved.
\end{proof}
Based on Theorem \ref{theorem}, we can find that the tag's location detection probability will be increased with the transmit power $P_s$ and reflection coefficient $|\Gamma|^2$ increase, but decreases as the distance between receive antenna $i$ and tag increases.
\subsection{SER Performance}
%In a constellation diagram, the minimum Euclidean distance $d_{min}$ between symbols directly impacts the error probability in a noise-only scenario, where
%\begin{equation}
%\begin{aligned}
%d_{\text{min}}=\min_{i\neq j}|\Gamma_i-\Gamma_j|.
%\end{aligned}
%\end{equation}
%In backscatter communication system, the received symbols are not only corrupted by additive noise but also multiplied by a cascaded fading amplitude that can vary each symbol's position in the constellation diagram. 
We consider the forward link channel $h_f$ and backscatter link channel of receiver $h_r$ both as Rayleigh fading,
% then the effective distance is 
%\begin{equation}\small
%\begin{aligned}
%d_{\text{eff}}=d_{\text{min}}|h_fh_r|.
%\end{aligned}
%\end{equation}
then the instantaneous defined SNR is
\begin{equation}
\begin{aligned}
{\gamma=\frac{\chi P_{s}G_{tag}G_TG_RL(d_f)L(d_r)\left|h_f h_r\Gamma\right|^2}{N_0}=\Xi|h_fh_r|^2,}
\end{aligned}
\end{equation}
where $G_R$ is the antenna gain of the receiver, $d_r$ is the distance between the receiver and tag, $N_0$ is the noise power spectral density, and $\Xi= \chi P_{s}G_{tag}G_TG_RL(d_f)L(d_r)\left| \Gamma\right|^2/N_0$.

We utilize the moment generation function (MGF)-based method to derive the SER. For the $M$-PSK and $M$-QAM modulation, the SER can be written as follows  \cite{BER}
\begin{equation}
\begin{aligned}\label{PSK}
P_{\text{MPSK}}=\frac{1}{\pi}\int_0^{\pi-\frac{\pi}{M}}\mathcal{M}_{\gamma}\left(\frac{g_{\text{MPSK}}}{\sin^2\theta}\right)d\theta,
\end{aligned}
\end{equation}
\begin{equation}\label{QAM}
\begin{aligned}
P_{\text{MQAM}}=\frac{4}{\pi}\left(1-\frac{1}{\sqrt{M}}\right)\int_0^{\frac{\pi}{2}}\mathcal{M}_{\gamma}\left(\frac{g_{\text{MQAM}}}{\sin^2\theta}\right)d\theta\notag\\
-\frac{4}{\pi}\left(1-\frac{1}{\sqrt{M}}\right)^2\int_0^{\frac{\pi}{4}}
\mathcal{M}_{\gamma}\left(\frac{g_{\text{MQAM}}}{\sin^2\theta}\right)d\theta,
\end{aligned}
\end{equation}
where $g_{\text{MPSK}}=\sin^2(\pi/M)$ and $g_{\text{MQAM}}=\frac{3}{2(M-1)}$ are the $M$-PSK and $M$-QAM constellation constant, respectively. $\mathcal{M}_{\gamma}(\cdot)$ is the MGF of  $\gamma$, defined as
\begin{equation}\label{MGFgammal}
\mathcal{M}_{\gamma}(s)=\mathbb{E}\{e^{-s\gamma}\}=\int_0^{\infty}e^{-s\gamma}f(\gamma)d{\gamma},
\end{equation}
where $f(\gamma)$ is the PDF of $\gamma$. Let $X_1=|h_f|^2$ and $X_2=|h_r|^2$, if $h_f$ and $h_r$ are both Rayleigh fading, we have $X_1$ and $X_2$ are both Gamma distribution, with PDF
\begin{align}\label{Fnaka1}
&f_{X_1}(x_1)=\frac{x_1}{\Omega_f^{2}}\exp\left(-\frac{x_1}{\Omega_f}\right),\\\label{Fnaka2}
&f_{X_2}(x_2)=\frac{x_2}{\Omega_r^2}\exp\left(-\frac{x_2}{\Omega_r}\right),
\end{align}
where $\Omega_f$ and $\Omega_r$ are Rayleigh fading parameters of the forward link and backscatter link, respectively.
The MGF of $X=X_1X_2$ is defined as
\begin{align}\label{MGF1}
\mathcal{M}_{X}(s)=\int_0^{\infty}\int_0^{\infty}e^{-x_1x_2s}f(x_1,x_2)dx_1dx_2.
\end{align}
Substituting \eqref{Fnaka1} and \eqref{Fnaka2} into \eqref{MGF1}, we can obtain \cite{BER}
\begin{align}\label{UnMGF}
\mathcal{M}_X(s)=\left(\frac{1}{\Omega_f\Omega_r s}\right)^2\Psi\left(2,1,\frac{1}{\Omega_f\Omega_rs}\right),
\end{align}
where  $\Psi(a,b;z)$ is the Tricomi confluent hypergeometric function also denoted by $U(a,b,z)$. Based on $\gamma=\Xi X$, we have
\begin{align}\label{MGFSNR}
\mathcal{M}_{\gamma}(s)=\mathcal{M}_X(\Xi s)=\left(\frac{1}{\Omega_f\Omega_r \Xi s}\right)^2\Psi\left(2,1,\frac{1}{\Omega_f\Omega_r\Xi s}\right).
\end{align}

Substituting \eqref{MGFSNR} into \eqref{PSK} and \eqref{QAM}, the average SER for $M$-PSK and $M$-QAM can be derived. However, the derivation contains the hypergeometric function, making the closed form of the SER difficult to obtain. Let $\theta=\pi/2$, we can derive the Chernoff bound as
\begin{align}\label{MPSK}
P_{\text{MPSK}}\leq \left(1-\frac{1}{{M}}\right)\left(\frac{1}{g'_{\text{MPSK}}}\right)^2\Psi\left(2,1,\frac{1 }{g'_{\text{MPSK}}}\right),\\\label{MQAM}
P_{\text{MQAM}}\leq \left(1-\frac{1}{{M}}\right)\left(\frac{1}{g'_{\text{MQAM}}}\right)^2\Psi\left(2,1,\frac{1 }{g'_{\text{MQAM}}}\right),
\end{align}
where $g'_{\text{MPSK}}=\sin^2(\pi/M)\Omega_f\Omega_r\Xi$, $g'_{\text{MQAM}}=\frac{3\Omega_f\Omega_r \Xi}{2(M-1)}$.
The monotonicity of the function is hard to obtain directly from the first derivative of  \eqref{MPSK} and \eqref{MQAM} due to the hypergeometric function $\Psi(a,b;z)$. We will use the numerical methods to evaluate the performance in the next section.
\section{Numerical Results}
In this section, we present the simulation results of the tag's location detection probability and SER. On one hand, the numerical results can verify the accuracy of the theoretical analysis.  On the other hand, they can provide an intuitive evaluation when the theoretical derivative is too complex or not feasible to compute. All the simulation parameter settings are shown in Table \ref{parameter}. We use MATLAB to generate the Rician and Rayleigh simulations, with the number of simulations set at $10^6$. The Meijer G-function and Tricomi confluent hypergeometric function are also available in MATLAB.
\begin{table}[t]
	\centering
	\caption{ Parameters Values}
	\label{parameter}
	\begin{tabular}{c|c|c}
		\specialrule{0em}{0.5pt}{0.5pt}
		\hline
		\rowcolor[gray]{.9}
		\small Parameter&\small Description&\small Value\\[1ex]
		\hline
		$\chi$&Polarization loss factor &0.5\\
		\hline $G_{tag}$ & Tag's antenna gain  & 0dBi\\
		\hline $G_T$ & AP's transmit antenna gain & 0dBi\\
		\hline $G_R^i$ & AP's receive antenna gain& 0dBi\\
		\hline $G_R$ & Receiver's antenna gain& 0dBi\\
		\hline $|\Gamma|^2$& Reflection coefficient& [0, 1]\\
		\hline $\eta$& Path loss factor  &1.8\\
		\hline $P_{th}$ &Sensitivity threshold &-75dBm\\
		\hline $\sigma_f$ & Forward link Rician parameter &1\\
		\hline $\sigma_b$ & Backscatter link Rician parameter &1\\
		\hline
	\end{tabular}
\end{table}

\reffig{Detection} illustrates the performance of tag's location detection probability under different fading conditions. The performance under Rayleigh fading is worse than that under Rician fading. This is because Rayleigh fading lacks a line-of-sight (LOS) component and is more severely affected by multipath propagation. If the theoretical analysis in Rician fading is hard to obtain, using Rayleigh fading as an alternative can provide a more conservative benchmark for system performance. Moreover, increasing the Rician fading parameters $K_f$ and $K_b$ improves the performance. It is interesting to note that the performance remains consistent whether $K_f=1, K_b=0$ or $K_f=0, K_b=1$, which suggests the impact of these parameters on the cascaded channel fading is symmetric.

In \reffig{Distance}, we demonstrate how the distance between the tag and the receive antennas affects the tag's location detection probability. As we can see, when the receive antenna distance is 5m away from the tag, the detection probability is about 30\% with $P_s=1W$. The detection probability significantly increases to  90\% when the distance is reduced to  2m. Moreover, the probability that at least 3 antennas can detect the tag's location is about 50\%. These results indicate that achieving a detection probability of at least 50\% requires that the distance between the tag and any receive antenna cannot exceed 5m. We also verified the independence of each receive antenna during simulations. The correlation heatmap in \reffig{correlation} shows the low correlation between different antennas, which supports the reliability of the simulation results.

In \reffig{recoefficient}, we analyze the effects of increasing the reflection coefficient on the tag's location detection probability and SER. The tag's location detection probability increases as the reflection coefficient increases, while the SER decreases. However, the variations in SER across different reflection coefficients are relatively minor, which indicates that the choice of the reflection coefficient is flexible without significantly compromising the quality of communication. Furthermore, the SER has better performance under QAM modulation than that of PSK.
\begin{figure}[t]
	\centering
	\includegraphics[width=75mm]{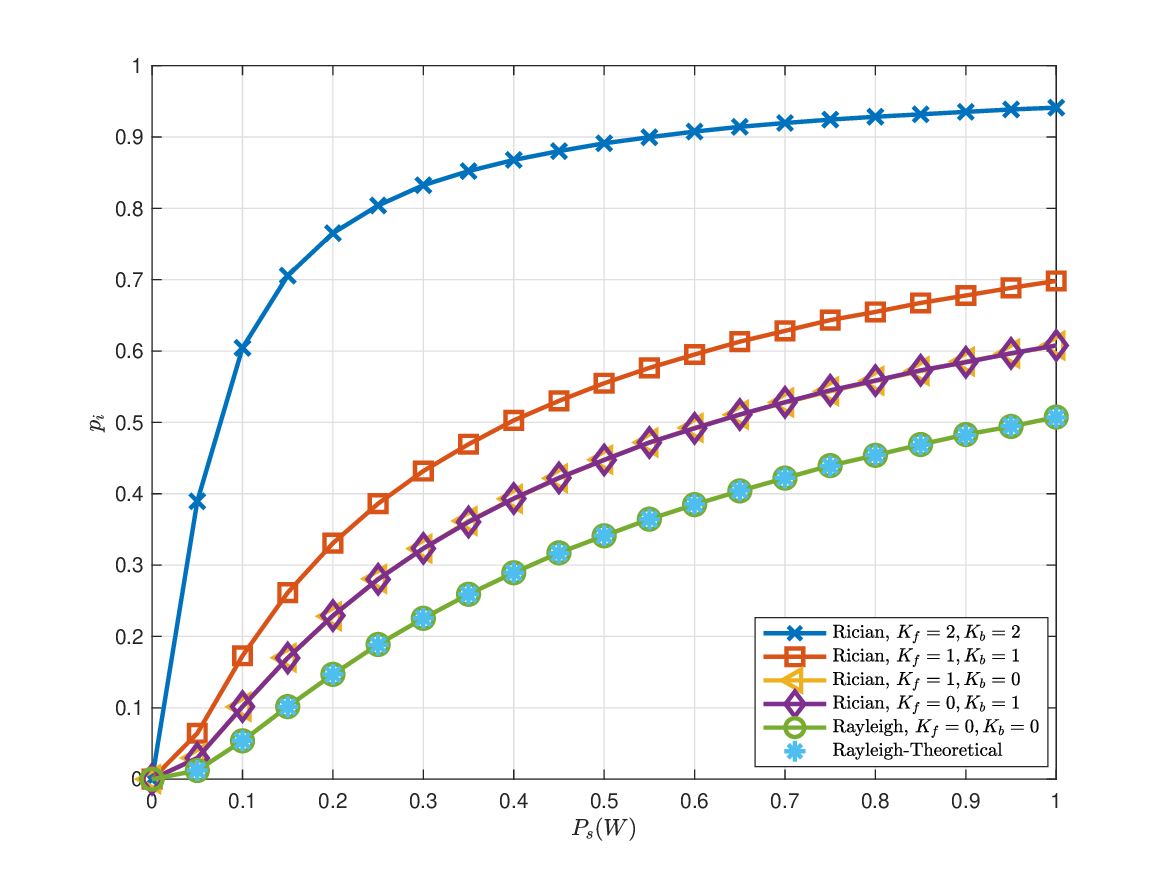}\\
	\caption{Tag's location detection probability of $i$-th receive antenna $p_i$ versus transmit power $P_s$ under different fading channels.}\label{Detection}
\end{figure}
\begin{figure}[t]
	\centering
	\includegraphics[width=75mm]{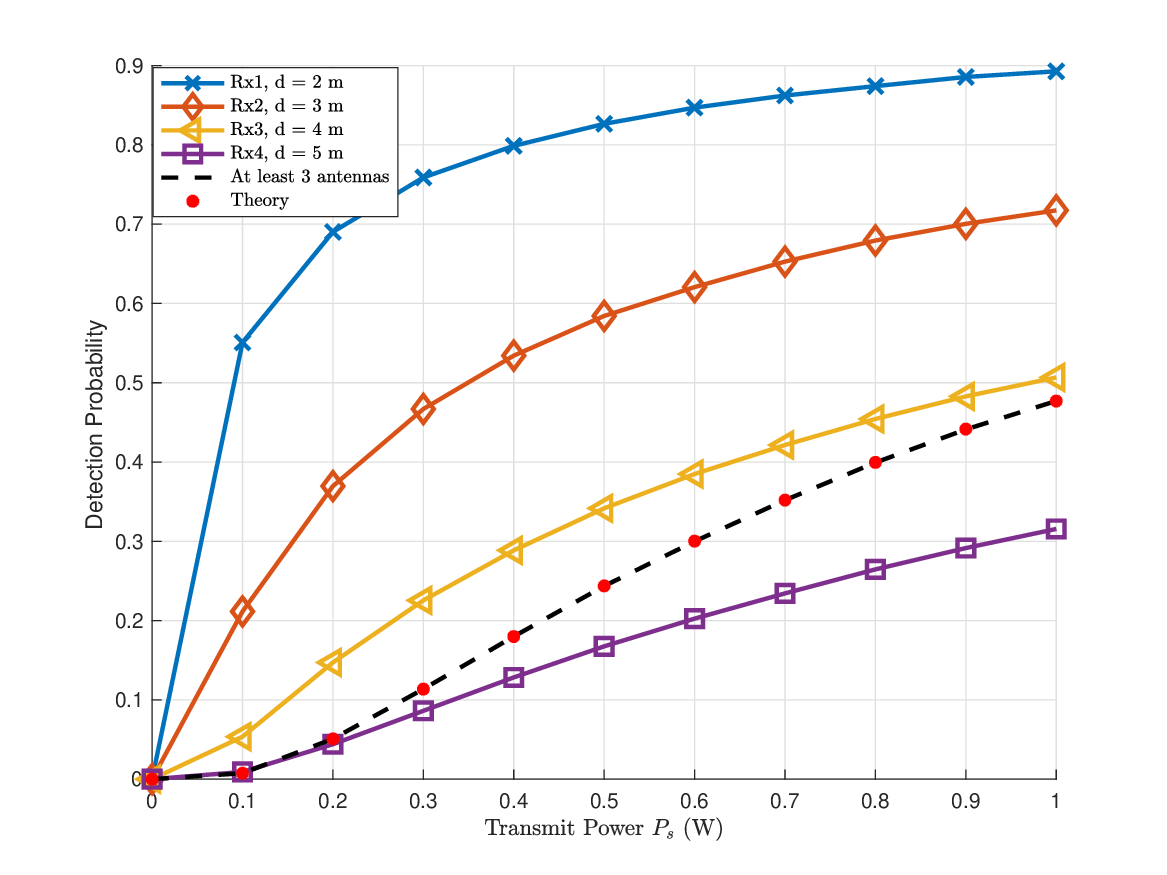}\\
	\caption{Tag's location detection probability for different receive antennas.}\label{Distance}
\end{figure}
\begin{figure}[t]
	\centering
	\includegraphics[width=75mm]{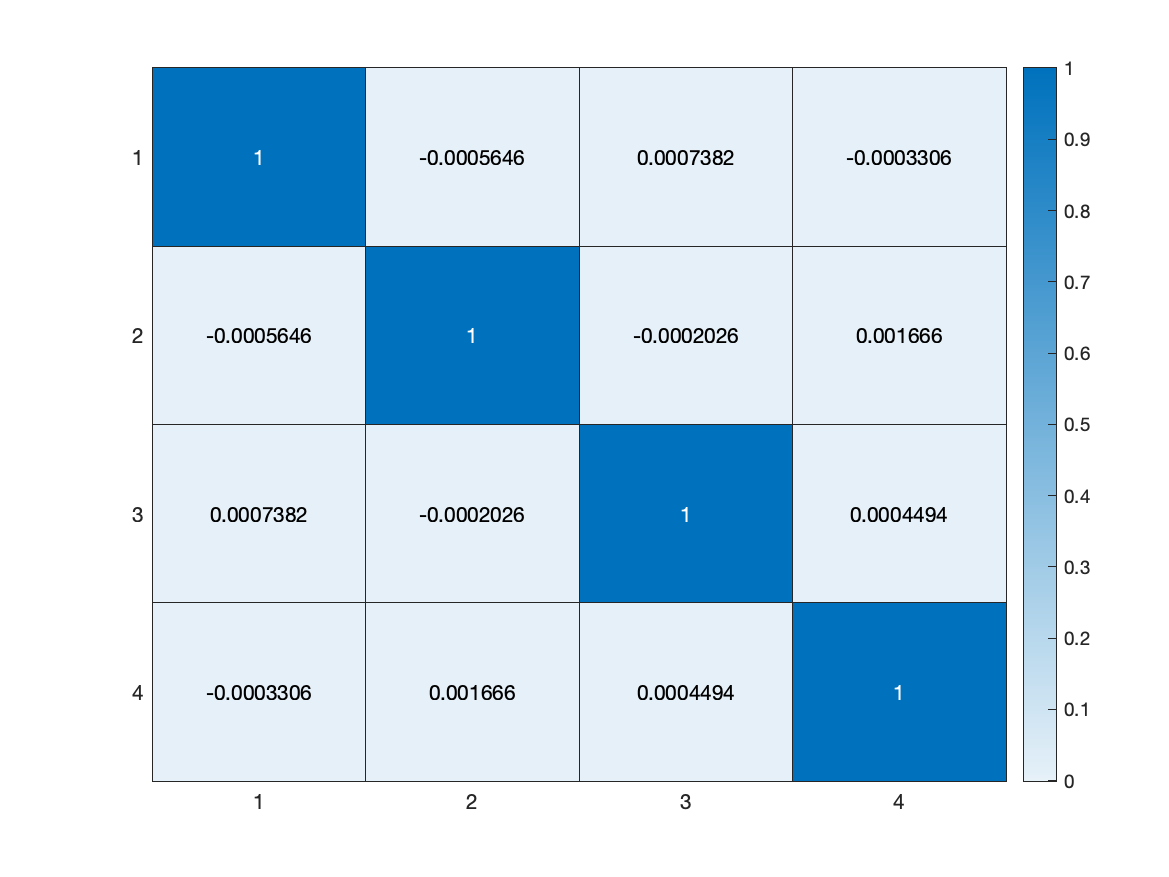}\\
	\caption{The correlation heatmap of different antennas.}\label{correlation}
\end{figure}
\begin{figure}[t]
	\centering
	\includegraphics[width=75mm]{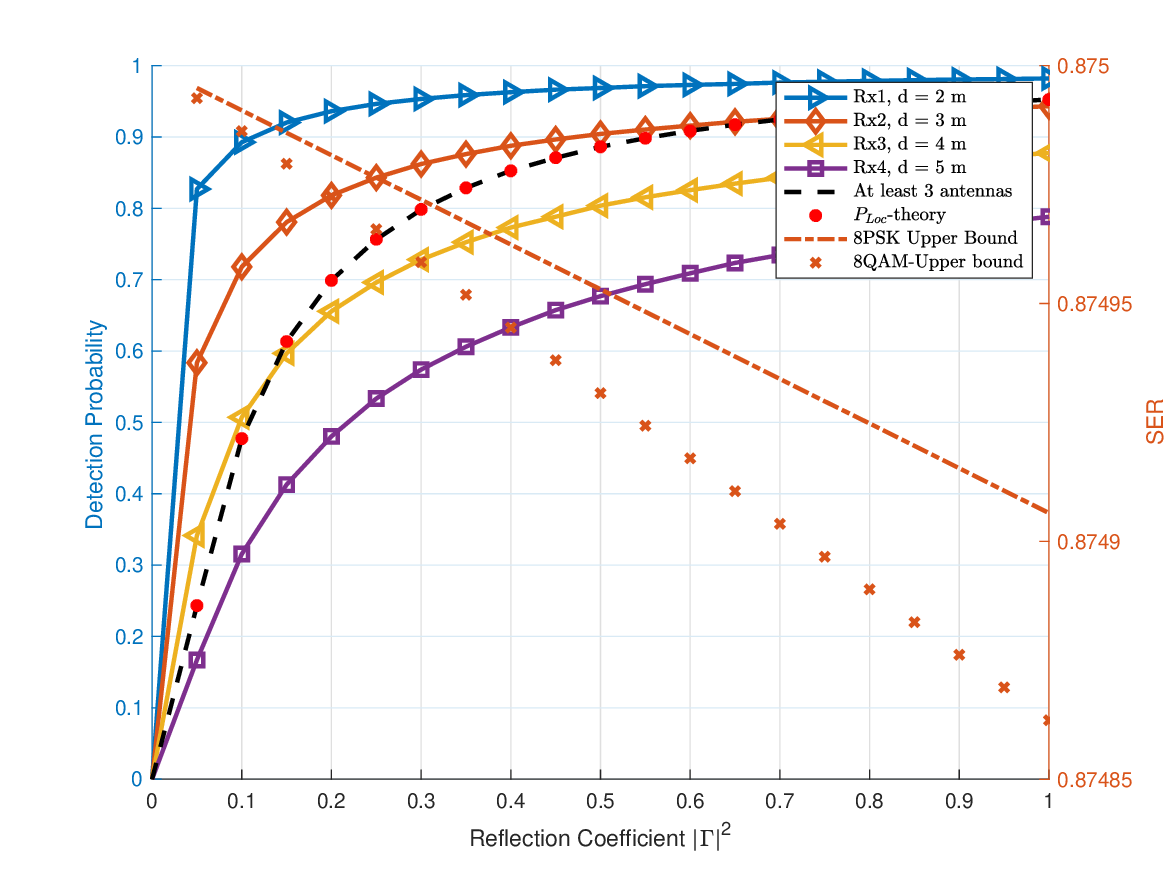}\\
	\caption{Tag's detection probability versus the reflection coefficient.}\label{recoefficient}
\end{figure}
\section{Conclusion}
In this paper, we propose a novel net-zero ISAC model utilizing the backscatter properties of net-zero devices, which simplifies the system design without additional power consumption and hardware. Our proposed frame structure design effectively mitigates interference at the transmitter, tag, and receiver. Through theoretical derivation of the tag's location detection probability and SER, we obtain valuable insights on parameter adjustments for improving sensing and communication performance. Furthermore, we demonstrate the flexibility of choosing the reflection coefficient to improve tag's location detection probability without significant sacrifice in communication performance. 
%Future research will explore more complex scenarios that accommodates higher densities of zero-power devices.
\section{Acknowledgement}
This work was supported by the UK Engineering and Physical Sciences Research Council (EPSRC) under Grant EP/Y000315/1.

\balance
 \end{document}